\documentclass[11pt]{article}
\usepackage{authblk,caption}
\usepackage{amsmath,amsthm,amssymb,mathtools,bm}
\usepackage{indentfirst,booktabs,microtype,xcolor}
\usepackage[
    backend=bibtex,
    style=alphabetic,
    giveninits=true,
    maxbibnames=99,
    minalphanames=3,
    date=year
]{biblatex}
\usepackage{hyperref}

\let\originalleft\left
\let\originalright\right
\renewcommand{\left}{\mathopen{}\mathclose\bgroup\originalleft}
\renewcommand{\right}{\aftergroup\egroup\originalright}

\hypersetup{
    breaklinks=true,
    colorlinks=true,
    citecolor=blue,
    urlcolor=blue,
    linkcolor=blue
}

\renewbibmacro{in:}{}
\AtEveryBibitem{\clearfield{issue}}

\newbibmacro{string+doiurl}[1]{%
    \iffieldundef{doi}{%
        \iffieldundef{url}{%
            #1%
        }{%
            \href{\thefield{url}}{#1}%
        }%
    }{%
        \href{https://doi.org/\thefield{doi}}{#1}%
    }%
}
\DeclareFieldFormat{title}{%
    \usebibmacro{string+doiurl}{\mkbibemph{#1}}%
}
\DeclareFieldFormat[
    article,thesis,incollection,inproceedings,manual
]{title}{%
    \usebibmacro{string+doiurl}{#1}%
}
\newcommand{\ket}[1]{\left| #1 \right>}
\newcommand{\bra}[1]{\left< #1 \right|}
\newcommand{\ketbra}[2]{\ket{#1} \! \bra{#2}}
\newcommand{\pure}[1]{\ketbra{#1}{#1}}
\newcommand{\tr}[2][]{\operatorname{Tr}_{#1}\left[#2\right]}
\newcommand{\Tr}{\operatorname{Tr}}
\newcommand{\binh}{h_{\mathrm{bin}}}
\newcommand{\defvar}{\coloneqq}
\newcommand{\dop}[1]{\operatorname{D}_{#1}}
\newcommand{\PPT}{\operatorname{PPT}}
\newcommand{\tcop}{\operatorname{T}_+}
\DeclareMathOperator{\SEP}{SEP}
\newcommand{\eps}{\varepsilon}
\newcommand{\id}{\mathbb{I}}
\newcommand{\idmap}{\operatorname{id}}
\newcommand{\norm}[1]{\left\lVert#1\right\rVert}
\newcommand{\CPTP}{\operatorname{CPTP}}
\newcommand{\supp}{\operatorname{supp}}
\newcommand{\term}[1]{\textup{\textbf{#1}}}
\newcommand{\bsym}[1]{\boldsymbol{#1}}
\newcommand{\bd}[1]{\mathcal{N}_{#1}}
\newcommand{\bdc}[1]{\mathcal{N}_{#1}^{c}}
\newcommand{\Pmap}{\mathcal{P}}
\newcommand{\Qt}[1]{Q_{\leftrightarrow}\left(#1\right)}
\newcommand{\Pt}[1]{P_{\leftrightarrow}\left(#1\right)}
\newcommand{\coh}[2]{I_c\left(\bsym{#1}_{#2}\right)}

\newtheorem{theorem}{Theorem}[section]
\newtheorem{lemma}{Lemma}[section]
\newtheorem{corollary}{Corollary}[section]
\theoremstyle{definition}
\newtheorem{definition}{Definition}[section]

\begin{document}

\title{\textbf{No advantage from classical assistance for\\
energy-constrained bosonic dephasing channels}}

\author[1]{Amir Arqand}
\author[2]{Graeme Smith}

\affil[1]{Institute for Quantum Computing and Department of Physics
and Astronomy, University of Waterloo\\
Waterloo, Ontario N2L 3G1, Canada.}

\affil[2]{Institute for Quantum Computing and Department of Applied
Mathematics, University of Waterloo\\
Waterloo, Ontario N2L 3G1, Canada.}

\date{}

\maketitle
\begin{abstract}
We prove that four capacities of single-mode bosonic dephasing channels coincide under an average photon-number constraint: the unassisted quantum capacity, the unassisted private capacity, the two-way classically assisted quantum capacity, and the two-way public-classically assisted private capacity. The assistance may be interactive: the parties can exchange classical messages between channel uses and choose subsequent inputs based on earlier outcomes. For every finite value of the constraint, their common value is the coherent information optimized over Fock-diagonal input states satisfying the constraint. The proof uses a relative entropy comparison with the completely dephasing channel to derive chain rules for the Rains relative entropy and the relative entropy of entanglement. These chain rules give an upper bound on the rates achievable by adaptive protocols whose mean input photon number is constrained on average over all channel uses.
We also show that both chain rules are attained by canonical purifications of finite-energy Fock-diagonal states.
\end{abstract}
\section{Introduction}
\label{sec:introduction}
The quantum and private capacities of a quantum channel quantify its
ability to transmit quantum information and private classical
information, respectively. These rates depend on the resources
available to the communicating parties. Allowing local operations and
two-way public classical communication (LOCC) can increase the rates
at which they distribute entanglement or establish a secret key.
The resulting protocols may be adaptive: the parties retain local
quantum memories and use the outcomes of earlier operations to choose
later inputs. Determining whether such assistance increases the capacities requires
accounting for correlations with the parties' quantum memories and
the intervening LOCC operations~\cite{PLOB17,DSW18}.

Photon loss and dephasing are two important noise processes affecting optical and microwave modes. Photon loss reduces the mean photon number, whereas dephasing preserves the photon-number distribution and reduces coherence between Fock states~\cite{LW23}. The pure-loss channel provides an example in which classical assistance changes the ability to communicate quantum information: its
unassisted quantum capacity vanishes at transmissivities at most
$1/2$, whereas its two-way assisted quantum capacity is positive
for every nonzero transmissivity~\cite{WPGG07,PLOB17}. 
Under an average photon-number constraint, the two-way quantum and secret-key capacities of the pure-loss channel were determined in~\cite{Pir26,FMLG26}. For the quantum-limited amplifier, Pirandola showed that the corresponding capacities coincide with the unassisted quantum and private capacities~\cite{Pir26}.

Following earlier studies of bosonic dephasing~\cite{JC10,AMM20},
Lami and Wilde obtained an exact expression for the quantum and
private capacities of general bosonic dephasing channels without an
input-energy constraint~\cite{LW23}. They showed that the unassisted
and two-way assisted capacities coincide and are equal to the
classical relative entropy of the phase-noise distribution with
respect to the uniform distribution. Their achievability argument
uses inputs uniformly distributed over increasingly many Fock levels,
whose mean photon number grows without bound. The equality of the
unrestricted capacities therefore does not by itself determine
whether assistance helps at fixed input energy.

For bosonic communication, a bound on the mean photon number is a
natural restriction on the allowed inputs. Wilde and Qi established
coding theorems for energy-constrained quantum and private
communication, including a single-letter characterization for
degradable channels satisfying suitable entropy conditions~\cite{WQ18}.
Davis, Shirokov, and Wilde developed the corresponding theory for
two-way assisted communication and derived upper bounds based on
energy-constrained squashed entanglement~\cite{DSW18}. For a Gaussian
phase-noise model, upper and lower bounds on the energy-constrained
LOCC-assisted quantum capacity were obtained in~\cite{AMM23}.
For general bosonic dephasing channels, the remaining assistance
question is whether adaptive protocols can exceed the unassisted
rate while satisfying the same average energy constraint.

In this work, we extend the equality between assisted and unassisted
capacities to finite mean-photon-number constraints. More precisely,
for a channel $\bd{\mu}$ defined
by a phase probability density and every $0\leq E<\infty$, we prove
\begin{align}
    \Qt{\bd{\mu},E}
    =\Pt{\bd{\mu},E}
    =Q(\bd{\mu},E)=P(\bd{\mu},E)=\sup_{\substack{\rho(k)\geq0,\ \sum_{k\geq0}\rho(k)=1\\
                       \sum_{k\geq0}k\rho(k)\leq E}}
          \coh{\rho}{A}.
    \label{eq:introduction-capacities}
\end{align}
Here $\bsym{\rho}_A=(\rho(k))_{k\geq0}$ is the photon-number
distribution, and $\coh{\rho}{A}$ is the coherent information of the
Fock-diagonal input $\rho_A=\sum_{k\geq0}\rho(k)\pure{k}_A$:
\begin{align}
    \coh{\rho}{A}
    =H(\bsym{\rho}_A)
       -H\!\left(\sum_{k\geq0}\rho(k)\pure{e_k}_F\right),
    \label{eq:introduction-coherent-information}
\end{align}
where $\{\ket{e_k}_F\}_k$ are the pure environment states associated
with the Fock inputs. The energy constraint is imposed on the mean
photon number averaged over all channel uses. Thus the result allows
adaptive protocols with different input energies in different rounds.
It gives a single-letter characterization of all four capacities in
terms of the same optimization over photon-number distributions.

The equality of the unassisted capacities follows from degradability
and the energy-constrained coding theorem of Wilde and
Qi~\cite{WQ18}. The main step is therefore to obtain a matching upper
bound for adaptive LOCC-assisted protocols. We begin with a relative
entropy bound comparing the output of the bosonic dephasing channel
with the output of the completely dephasing channel. Combining this
bound with a pinching identity gives a relative entropy chain rule
(Theorem~\ref{thrm:EC chain rule}). We then use this result to show
that the increase in either the Rains relative entropy or the relative
entropy of entanglement across one channel use is at most
$\coh{\rho}{A}$, where $\bsym{\rho}_A$ is the photon-number
distribution induced by the input marginal $\rho_A$
(Theorems~\ref{lem:rains-increment} and~\ref{thrm:REE chain rule}).

These bounds apply even when the channel input is correlated with the
parties' local memories. Applying them successively, together with
monotonicity under the intervening LOCC maps and concavity of the
coherent information, yields the required converse under the average
energy constraint. Furthermore, both entanglement bounds are attained
by the canonical purifications of finite-energy Fock-diagonal inputs
(Corollary~\ref{cor:entanglement-chain-rule-tightness}). The argument
allows countably supported photon-number distributions and does not
require a finite-dimensional restriction on the local memories or
private-state shields.

The structure of the paper is as follows. In
Sec.~\ref{sec:notation}, we introduce the notation, the bosonic
dephasing channel, and the energy-constrained capacities. Section~\ref{sec:EE QP}
contains the main results: Sec.~\ref{sec:rel ent chain} proves a relative entropy bound
for the channel output, Sec.~\ref{sec:QC} derives the chain rules and the quantum capacity
theorem, and Sec.~\ref{sec:private-capacities} establishes the private
capacity theorem and tightness of the entanglement bounds.
Section~\ref{sec:conclusion} concludes with open questions.
\section{Preliminaries}
\label{sec:notation}

Each system $A$ is associated with a separable Hilbert space
$\mathcal H_A$, which may be finite- or infinite-dimensional.
Composite systems are described by tensor products, so that
$\mathcal H_{AB}=\mathcal H_A\otimes\mathcal H_B$.
States are positive trace-class operators with unit trace, and
channels act on trace-class operators. All logarithms are base~$2$.
We use the notation in Table~\ref{tab:notation}.

\begin{table}[h!]
\centering
\caption{List of notation}
\label{tab:notation}
\renewcommand{\arraystretch}{1.25}
\begin{tabular}{@{}p{0.29\linewidth}p{0.66\linewidth}@{}}
\toprule
\textit{Symbol} & \textit{Definition} \\
\midrule
$H(A)_\rho$, $H(\bsym p)$
& Von Neumann entropy of $\rho_A$ and Shannon entropy of $\bsym p$ \\
\addlinespace
$X\geq Y$
& $X-Y$ is positive semidefinite \\
\addlinespace
$\tcop(A)$
& Positive trace-class operators on $\mathcal H_A$ \\
\addlinespace
$\dop{=}(A)$
& $\{\rho\in\tcop(A):\Tr\rho=1\}$ \\
\addlinespace
$\CPTP(A,B)$
& Completely positive, trace-preserving maps from $A$ to $B$ \\
\addlinespace
$\id_A$, $\idmap_A$
& Identity operator and identity channel on $A$, respectively \\
\addlinespace
$\norm{X}_1$, $\norm{X}_\infty$
& Trace norm $\Tr\sqrt{X^\dagger X}$ and operator norm, respectively \\
\addlinespace
$\supp(X)$
& Support of a positive operator $X$ \\
\addlinespace
$A_j^k$
& Composite register $A_j\cdots A_k$ \\
\bottomrule
\end{tabular}
\end{table}

Tensor products with identity channels are often left implicit.
For $\mathcal E\in\CPTP(A,B)$, we write
\begin{align}
    \mathcal E[\rho_{AR}]
    \defvar (\mathcal E\otimes\idmap_R)[\rho_{AR}].
\end{align}
The same convention applies to completely positive maps.
\begin{definition}[Classical--quantum states]
\label{def:classical-quantum}
A state $\rho_{CQ}\in\dop{=}(CQ)$ is \term{classical on $C$}
with respect to a fixed orthonormal basis $\{\ket{c}_C\}_c$ if
\begin{align}
    \rho_{CQ}=\sum_c\rho(c)\pure{c}_C\otimes\rho_{Q|c},
    \label{eq:cq}
\end{align}
where $(\rho(c))_c$ is a probability distribution and
$\rho_{Q|c}\in\dop{=}(Q)$ are the conditional states, chosen
arbitrarily when $\rho(c)=0$. The basis is understood from context.
\end{definition}

More generally, for a state $\rho_A\in\dop{=}(A)$ and a fixed
orthonormal basis $\{\ket{k}_A\}_k$, write
\begin{align}
    \rho(k)\defvar\bra{k}\rho_A\ket{k},
    \qquad
    \bsym{\rho}_A\defvar(\rho(k))_k.
    \label{eq:measurement-distribution}
\end{align}
This is the distribution obtained by measuring $\rho_A$ in that basis;
$\rho_A$ need not be diagonal. For the corresponding completely
dephasing channel $\Delta_A$,
\begin{align}
    \Delta_A[\rho_A]=\sum_k\rho(k)\pure{k}_A,
    \qquad
    H(\bsym{\rho}_A)=H(A)_{\Delta_A[\rho_A]},
\end{align}
where $H(\bsym{\rho}_A)$ denotes Shannon entropy.
For bosonic systems, unless stated otherwise, we use the Fock basis,
with $k\geq0$, so that
$\bsym{\rho}_A=(\rho(0),\rho(1),\ldots)$.

\begin{definition}[Umegaki relative entropy~\cite{Ume62,LS23}]
\label{def:relative-entropy}
For $\rho\in\dop{=}(A)$ and $\sigma\in\tcop(A)$, define
\begin{align}
 D(\rho\Vert\sigma)\defvar
 \begin{cases}
  \Tr[\rho(\log\rho-\log\sigma)],
       &\supp(\rho)\subseteq\supp(\sigma),\\
  +\infty,&\text{otherwise}.
 \end{cases}
 \label{eq:umegaki_div}
\end{align}
The trace expression is understood in the standard extended-real sense
for infinite-dimensional systems, rather than as a difference of two
possibly infinite traces. For classical probability distributions
$\bsym p,\bsym q$, the Kullback--Leibler divergence
$D(\bsym p\Vert\bsym q)$ is obtained by identifying the distributions
with diagonal states.
\end{definition}

\begin{definition}[Rains relative entropy~\cite{Rai01,ADMVW02}]
\label{def:rains-relative-entropy}
For $\rho\in\dop{=}(AB)$, define
\begin{align}
 R(A:B)_\rho\defvar\inf_{\sigma\in\PPT'(A:B)}
 D(\rho_{AB}\Vert\sigma_{AB}),
 \label{eq:rains-relative-entropy}
\end{align}
where
\begin{align}
 \PPT'(A:B)\defvar
 \left\{\sigma\in\tcop(AB):\norm{\sigma^{T_B}}_1\leq1\right\}.
 \label{eq:rains-feasible-set}
\end{align}
Here $T_B$ denotes partial transposition in a fixed orthonormal basis
of $B$. It is defined on Hilbert--Schmidt operators, with
$\norm{\sigma^{T_B}}_1=+\infty$ if the result is not trace class.
The value of $R$ is independent of the choice of this basis.
\end{definition}

\begin{definition}[Relative entropy of entanglement~\cite{VPRK97,LS23}]
\label{def:relative-entropy-of-entanglement}
For a state $\rho_{AB}\in\mathrm{D}_{=}(AB)$, the relative entropy of
entanglement is defined as
\begin{equation}
  E_R(A:B)_\rho
    :=\inf_{\sigma_{AB}\in\SEP(A:B)}
         D(\rho_{AB}\|\sigma_{AB}),
  \label{eq:relative-entropy-of-entanglement}
\end{equation}
where $\SEP(A:B)$ is the trace-norm-closed convex hull of product states
$\sigma_A\otimes\sigma_B$ with $\sigma_A\in\mathrm{D}_{=}(A)$ and
$\sigma_B\in\mathrm{D}_{=}(B)$, and $D$ denotes the quantum relative entropy.
\end{definition}

\subsection{Energy-constrained unassisted capacities}
\label{sec:unassisted-capacities}

\begin{definition}\label{def:energy-constraint}
(Energy constraint)
For a single-mode bosonic input, let
\begin{align}
    \hat n\defvar\sum_{k=0}^{\infty}k\pure{k},
    \qquad
    G_n\defvar\frac1n\sum_{i=1}^{n}
       \id_A^{\otimes(i-1)}\otimes\hat n\otimes\id_A^{\otimes(n-i)}.
    \label{eq:block-energy-observable}
\end{align}
An $n$-use input $\rho_{A_1^n}\in\dop{=}(A_1^n)$ satisfies the
energy constraint $E\geq0$ if
\begin{align}
    \Tr[G_n\rho_{A_1^n}]
       =\frac1n\sum_{i=1}^{n}\Tr[\hat n\rho_{A_i}]\leq E.
    \label{eq:block-energy}
\end{align}
\end{definition}

\begin{definition}\label{def:unassisted-capacities}
(Unassisted capacities~\cite[Sec.~III]{WQ18})
Let $\mathcal N\in\CPTP(A,B)$ be a memoryless channel.
We write
\begin{align}
    Q(\mathcal N,E)\defvar Q(\mathcal N,\hat n,E),\qquad
    P(\mathcal N,E)\defvar P(\mathcal N,\hat n,E)
\end{align}
for the supremum of the achievable vanishing-error rates of quantum
and private classical communication, respectively, under
\eqref{eq:block-energy}. The constraint holds for every encoded logical
input or message, respectively. The equivalent average-input
formulations are given in~\cite[Theorem~1]{WQ18}.
\end{definition}

\begin{definition}\label{def:coherent-private-information}
(Coherent and private information)
Let $\mathcal N^c\in\CPTP(A,F)$ be complementary to
$\mathcal N\in\CPTP(A,B)$. When the output entropies are finite, set
\begin{align}
    I_c(\rho_A,\mathcal N)\defvar
       H(B)_{\mathcal N[\rho_A]}-H(F)_{\mathcal N^c[\rho_A]}.
    \label{eq:coherent-information}
\end{align}
For a finite or countable ensemble $\{p(x),\rho_A^x\}_x$, put
$\xi_{XA}\defvar\sum_xp(x)\pure{x}\otimes\rho_A^x$,
$\omega_{XB}\defvar\mathcal N[\xi_{XA}]$, and
$\omega_{XF}\defvar\mathcal N^c[\xi_{XA}]$, with identity channels
implicit. {Here $I(X:Y)_\omega\defvar
D(\omega_{XY}\Vert\omega_X\otimes\omega_Y)$ is the mutual information.
For $0\leq E<\infty$, define}
\begin{align}
    C_p(\mathcal N,\hat n,E)\defvar
       \sup_{\substack{\{p(x),\rho_A^x\}_x\\
           \Tr[\hat n\sum_xp(x)\rho_A^x]\leq E}}
          \bigl\{I(X:B)_\omega-I(X:F)_\omega\bigr\}.
    \label{eq:private-information}
\end{align}
\end{definition}

\begin{lemma}\label{lem:degradable-capacities}
(Energy-constrained capacities of degradable channels~\cite[Theorem~3]{WQ18})
Let $\mathcal N\in\CPTP(A,B)$ have a single-mode bosonic input,
$E\geq0$, and $\mathcal N^c=\mathcal D\circ\mathcal N$ for some
$\mathcal D\in\CPTP(B,F)$. Suppose that
\begin{align}
    \sup_{\rho_A\in\dop{=}(A):\,\Tr[\hat n\rho_A]\leq E}
         H(B)_{\mathcal N[\rho_A]}<\infty.
\end{align}
Then
\begin{align}
    Q(\mathcal N,E)=P(\mathcal N,E)
       =\sup_{\rho_A\in\dop{=}(A):\,\Tr[\hat n\rho_A]\leq E}
            I_c(\rho_A,\mathcal N)
       =C_p(\mathcal N,\hat n,E).
    \label{eq:degradable-capacities}
\end{align}
\end{lemma}
\begin{proof}
The capacity equality follows from~\cite[Theorem~3]{WQ18}, since
$\Tr[2^{-\beta\hat n}]=(1-2^{-\beta})^{-1}<\infty$ for $\beta>0$.
For a degradable channel, pure-state refinement cannot decrease private
information, and the private information of a pure-state ensemble equals
the coherent information of its average state. This gives the equality
with $C_p$.
\end{proof}

\subsection{Bosonic dephasing channels}
\label{sec:bosonic-dephasing}

\begin{definition}\label{def:bosonic-dephasing}
(Bosonic dephasing channel~\cite[Eqs.~(1)--(3)]{LW23})
Let $\mu$ be a probability distribution on $[-\pi,\pi]$ with density $p$, so that
$p(\phi)\geq0$ and $\int_{-\pi}^{\pi}p(\phi)\,d\phi=1$.
The bosonic dephasing channel $\bd{\mu}\in\CPTP(A,B)$ is
\begin{align}
    \bd{\mu}[\rho_A]
       &\defvar\int_{-\pi}^{\pi}p(\phi)
          e^{-i\phi\hat n}\rho_A e^{i\phi\hat n}\,d\phi
          \label{eq:bdc-integral}\\
       &=\sum_{k,\ell\geq0}\rho_{k\ell}(T_\mu)_{k\ell}
                            \ketbra{k}{\ell}_B,
          \label{eq:bdc-fock-action}
\end{align}
where $\hat n\defvar\sum_{k\geq0}k\pure{k}$ is the photon number
operator, $\rho_{k\ell}\defvar\bra{k}\rho_A\ket{\ell}$, and
\begin{align}
    (T_\mu)_{k\ell}
       \defvar\int_{-\pi}^{\pi}p(\phi)e^{-i\phi(k-\ell)}\,d\phi.
       \label{eq:bdc-toeplitz}
\end{align}
Here the input and output Fock spaces are identified; $\bd{\mu}$
is denoted by $\bd{p}$ in~\cite{LW23}.
A complementary channel $\bdc{\mu}\in\CPTP(A,F)$ is
\begin{align}
    \bdc{\mu}[\rho_A]
       =\sum_{k\geq0}\rho(k)\pure{e_k}_F,
       \qquad \rho(k)=\rho_{kk},
       \label{eq:bdc-complement}
\end{align}
where the normalized environment vectors $\ket{e_k}_F$ satisfy
\begin{align}
    \langle e_\ell|e_k\rangle=(T_\mu)_{k\ell}.
    \label{eq:environment-overlaps}
\end{align}
\end{definition}
\begin{lemma}
    \label{lemma:BD and CD}
    Let $\bd{\mu}\in\CPTP(A,B)$ be a bosonic dephasing channel, $\Delta\in\CPTP(A,B)$, $\Delta_A\in\CPTP(A,A)$, and $\Delta_B\in\CPTP(B,B)$ 
be completely dephasing channels in the Fock bases. Then,
    \begin{align}
        \label{eq:BD and CD}
        \Delta_B\circ\bd{\mu}=\bd{\mu}\circ\Delta_A=\Delta.
    \end{align}
\end{lemma}
\begin{proof}
Phase rotations preserve diagonal matrix elements in the Fock basis and fix each Fock
projector. Averaging over $\mu$ therefore gives, for every
$\rho_A\in\dop{=}(A)$,
\begin{align*}
    (\Delta_B\circ\bd{\mu})[\rho_A]
    &=\sum_{k\geq0}\bra{k}\rho_A\ket{k}\,\pure{k}_B\\
    &=(\bd{\mu}\circ\Delta_A)[\rho_A]
     =\Delta[\rho_A].\qedhere
\end{align*}
\end{proof}
\begin{lemma}\label{lemma:finiteness}
    Let $\bd{\mu}\in\CPTP(A,B)$ be a bosonic dephasing channel, and let $\rho_A\in\dop{=}(A)$ satisfy
    \begin{align}
        \Tr[\hat{n}\rho_A]\leq E<\infty,
    \end{align}
    where $\hat{n}$ is the number operator. Then, we have
    \begin{align}\label{eq:finiteness}
        \max\left\{H(A)_\rho,H(B)_{\bd{\mu}[\rho_A]},H(F)_{\bdc{\mu}[\rho_A]} \right\}\leq g(E)<\infty,
    \end{align}
    where 
    \begin{align}\label{eq:g(e)}
        g(E)\defvar\begin{cases}
            (E+1)\log(E+1)-E\log E& E>0\\
            0&E=0.
        \end{cases}
    \end{align}
\end{lemma}
\begin{proof}
    Let $\Delta_A\in\CPTP(A,A)$ and $\Delta_B\in\CPTP(B,B)$ be the completely dephasing channels. First observe that 
    \begin{align}
        \Tr\left[\hat{n}\Delta_A[\rho_A]\right]=\Tr[\hat{n}\rho_A]\leq E.
    \end{align}
    Since $\Delta_A$ is unital, data processing under unital channels~\cite[Theorems~2.2 and~3.3]{LB13} gives
    \begin{align}\label{eq:DPIunital}
        H(A)_\rho\leq H(A)_{\Delta_A[\rho_A]}.
    \end{align}
    For any $\beta>0$, let
    \begin{align}
        \theta_\beta\defvar\frac{2^{-\beta\hat{n}}}{\Tr[2^{-\beta\hat{n}}]}.
    \end{align}
    The Gibbs variational bound~\cite[Eqs.~(160)--(161)]{WQ18} gives
    \begin{align}\label{eq:H(A)}
        H(A)_{\Delta_A[\rho_A]}\leq\beta E+\log\Tr[2^{-\beta\hat{n}}]=\beta E-\log(1-2^{-\beta}).
    \end{align}
    Since the above relation holds for any $\beta>0$, for $E>0$, set
    \begin{align}
        \beta=\log(1+\frac{1}{E}),
    \end{align}
    then, Eq.~\eqref{eq:H(A)} becomes
    \begin{align}\label{eq:H(A) to g}
        H(A)_{\Delta_A[\rho_A]}\leq(E+1)\log(E+1)-E\log E=g(E).
    \end{align}
    For $E=0$, taking $\beta\to\infty$ in~\eqref{eq:H(A)} gives $H(A)_{\Delta_A[\rho_A]}=0$. In either case, combining the resulting bound with~\eqref{eq:DPIunital} yields
    \begin{align}
        H(A)_\rho\leq g(E).
    \end{align}
    The channel $\bd{\mu}$ also preserves mean photon number:
    \begin{align}
         \Tr\left[\hat{n}\bd{\mu}[\rho_A]\right]=\Tr[\hat{n}\rho_A]\leq E.
    \end{align}
    Then
    \begin{align}
        H(B)_{\bd{\mu}[\rho_A]}&\leq H(B)_{(\Delta_B\circ\bd{\mu})[\rho_A]}\nonumber
        \\
        &=H(A)_{\Delta_A[\rho_A]}\nonumber\\
        &\leq g(E),
    \end{align}
    The first inequality again follows from unitality of $\Delta_B$, the equality from~\eqref{eq:BD and CD}, and the last inequality from the bound just established. Finally, the complementary output is
    \begin{align}
        \bdc{\mu}[\rho_A]=\sum_{k\geq0}\rho(k)\ketbra{e_k}{e_k}_F,
    \end{align}
    for some normalized states $\{\ket{e_k}_F\}_{k\geq0}$. Let $\bsym{\rho}_A$ be the distribution induced by $\rho_A$ in the Fock basis, then,
    \begin{align}
        H(F)_{\bdc{\mu}[\rho_A]}\leq H(\bsym{\rho}_A)=H(A)_{\Delta_A[\rho_A]}\leq g(E).
    \end{align}
\end{proof}
For a fixed channel $\bd{\mu}$ and finite mean photon number, write
\begin{align}
    \coh{\rho}{A}
       \defvar I_c(\Delta_A[\rho_A],\bd{\mu})
       =H(\bsym{\rho}_A)
          -H\!\left(\sum_{k\geq0}\rho(k)\pure{e_k}_F\right).
    \label{eq:diagonal-coherent-information}
\end{align}

\begin{corollary}\label{cor:bdc-unassisted-capacities}
(Unassisted capacities of bosonic dephasing channels)
Let $\bd{\mu}\in\CPTP(A,B)$ and ${0\leq E<\infty}$.
For every $\rho_A\in\dop{=}(A)$ with $\Tr[\hat n\rho_A]\leq E$,
\begin{align}
    I_c(\rho_A,\bd{\mu})
       \leq\coh{\rho}{A}.
       \label{eq:fock-diagonal-reduction}
\end{align}
Consequently,
\begin{align}
    Q(\bd{\mu},E)=P(\bd{\mu},E)
       =\sup_{\substack{\rho(k)\geq0,\ \sum_{k\geq0}\rho(k)=1\\
                        \sum_{k\geq0}k\rho(k)\leq E}}
          \coh{\rho}{A}.
       \label{eq:bdc-unassisted-capacities}
\end{align}
\end{corollary}
\begin{proof}
The channel is degradable via
$\mathcal D[\tau_B]\defvar\sum_k\bra{k}\tau_B\ket{k}\pure{e_k}_F$.
Lemma~\ref{lemma:finiteness} verifies the entropy assumption in
Lemma~\ref{lem:degradable-capacities}.
By Lemma~\ref{lemma:BD and CD}, the identity
$\bdc{\mu}\circ\Delta_A=\bdc{\mu}$, and entropy increase under
Fock dephasing,
\begin{align}
    I_c(\rho_A,\bd{\mu})
       \leq H(A)_{\Delta_A[\rho_A]}-H(F)_{\bdc{\mu}[\rho_A]}
       =\coh{\rho}{A}.
\end{align}
Since $\Delta_A$ preserves mean photon number,
Lemma~\ref{lem:degradable-capacities} gives
\eqref{eq:bdc-unassisted-capacities}. The Fock-state ensemble
$\{\rho(k),\pure{k}\}_k$ also has private information
$\coh{\rho}{A}$.
\end{proof}

\subsection{Two-way LOCC-assisted energy-constrained capacities}
\label{sec:assisted-capacities}

We use the squared fidelity
$F(\rho,\sigma)\defvar\norm{\sqrt\rho\sqrt\sigma}_1^2$.
The following definitions specialize~\cite[Sec.~IV]{DSW18} to
$G=\hat n$ and energy parameter $P=E$.

\begin{definition}\label{def:adaptive-protocol}
(LOCC-assisted protocol and energy constraint)
An $n$-use protocol for $\mathcal N\in\CPTP(A,B)$ consists of local
initial states, $n$ channel uses, and countably decomposable LOCC maps
$\mathcal L^{(1)},\ldots,\mathcal L^{(n+1)}$.
LOCC means local operations and unlimited two-way public classical
communication. The first map prepares
$\rho^{(1)}_{A'_1A_1B'_1}$, separable across $A'_1A_1:B'_1$.
The subsequent states satisfy
\begin{align}
    \sigma^{(i)}_{A'_iB_iB'_i}
       &\defvar(\idmap_{A'_i}\otimes\mathcal N_{A_i\to B_i}
                           \otimes\idmap_{B'_i})
                    [\rho^{(i)}_{A'_iA_iB'_i}],
       &&i=1,\ldots,n,\label{eq:adaptive-channel-use}\\
    \rho^{(i+1)}_{A'_{i+1}A_{i+1}B'_{i+1}}
       &\defvar\mathcal L^{(i+1)}[\sigma^{(i)}_{A'_iB_iB'_i}],
       &&i=1,\ldots,n-1.\label{eq:adaptive-locc}
\end{align}
Here Alice holds $A'_iA_i$ before use $i$, Bob holds $B'_i$, and
$\omega\defvar\mathcal L^{(n+1)}[\sigma^{(n)}]$ is the final state.
The protocol satisfies the energy constraint $E$ if
\begin{align}
    \frac1n\sum_{i=1}^{n}\Tr[\hat n\rho^{(i)}_{A_i}]\leq E,
    \qquad
    \rho^{(i)}_{A_i}\defvar
         \tr[A'_iB'_i]{\rho^{(i)}_{A'_iA_iB'_i}}.
    \label{eq:adaptive-energy}
\end{align}
The input marginals are unconditional, averaged over the public
transcript and local randomness.
\end{definition}

\begin{definition}\label{def:assisted-quantum}
(Two-way LOCC-assisted energy-constrained quantum capacity)
An $(n,K,\hat n,E,\eps)$ quantum communication protocol is a
protocol from Definition~\ref{def:adaptive-protocol} whose final state
satisfies
\begin{align}
    F(\omega_{K_AK_B},\Phi^{(K)}_{K_AK_B})\geq1-\eps,
    \label{eq:assisted-entanglement-fidelity}
\end{align}
where $\dim K_A=\dim K_B=K$ and
\begin{align}
    \ket{\Phi^{(K)}}_{K_AK_B}\defvar\frac1{\sqrt K}
          \sum_{k=1}^{K}\ket{k}_{K_A}\otimes\ket{k}_{K_B},
    \qquad
    \Phi^{(K)}\defvar\pure{\Phi^{(K)}}.
\end{align}
A rate $r\geq0$ is achievable if, for every $\eps\in(0,1)$ and
$\delta>0$, such protocols exist for all sufficiently large $n$ with
$n^{-1}\log K\geq r-\delta$. Define
\begin{align}
    \Qt{\mathcal N,E}\defvar \Qt{\mathcal N,\hat n,E}
       \defvar\sup\{r\geq0:r\text{ is achievable}\}.
    \label{eq:assisted-quantum-capacity}
\end{align}
\end{definition}

\begin{definition}\label{def:private-state}
(Private state~\cite[Sec.~II~F]{DSW18})
A private state with key dimension $K$ has the form
\begin{align}
    \gamma_{K_AS_AK_BS_B}
       &\defvar U\bigl(\Phi^{(K)}_{K_AK_B}\otimes\tau_{S_AS_B}\bigr)U^\dagger,
       \label{eq:private-state}\\
    U&\defvar\sum_{j,k=1}^{K}\pure{j}_{K_A}\otimes\pure{k}_{K_B}
                         \otimes U^{jk}_{S_AS_B},
       \label{eq:twisting-unitary}
\end{align}
where $\tau\in\dop{=}(S_AS_B)$ and each $U^{jk}$ is unitary.
The shield systems $S_A,S_B$ are retained by Alice and Bob and may be
infinite dimensional. Measuring $K_A,K_B$ in the indicated bases gives
a uniform shared key independent of every purifying system.
\end{definition}

\begin{definition}\label{def:assisted-private}
(Two-way LOCC-assisted energy-constrained private capacity)
An $(n,K,\hat n,E,\eps)$ secret-key-agreement protocol is a
protocol from Definition~\ref{def:adaptive-protocol} whose final state,
with local shield systems retained, satisfies
\begin{align}
    F(\omega_{K_AS_AK_BS_B},\gamma_{K_AS_AK_BS_B})\geq1-\eps
    \label{eq:assisted-private-fidelity}
\end{align}
for a private state of key dimension $K$.
With achievable rates defined as in Definition~\ref{def:assisted-quantum},
using these protocols, set
\begin{align}
    \Pt{\mathcal N,E}\defvar \Pt{\mathcal N,\hat n,E}
       \defvar\sup\{r\geq0:r\text{ is achievable}\}.
    \label{eq:assisted-private-capacity}
\end{align}
This is also called the energy-constrained secret-key-agreement capacity.
Equivalently, after measuring the keys, the state
$\zeta_{K_AK_BF_1^nC}$ satisfies
\begin{align}
    F\!\left(\zeta_{K_AK_BF_1^nC},
       \frac1K\sum_{k=1}^K\pure{k}_{K_A}\otimes\pure{k}_{K_B}
                        \otimes\theta_{F_1^nC}\right)\geq1-\eps
    \label{eq:secret-key-security}
\end{align}
for some $\theta\in\dop{=}(F_1^nC)$, where $F_1^n$ are the channel
environments and $C$ is the complete public transcript held by Eve.
This equivalence uses a purified implementation retaining local
shield systems~\cite[Sec.~IV~B]{DSW18}.
\end{definition}

\section{Energy-constrained quantum and private capacities}
\label{sec:EE QP}
\subsection{A relative entropy bound for bosonic dephasing channels}
\label{sec:rel ent chain}
\begin{lemma}\label{lemma:div to Ic}
    Let $\bd{\mu}\in\CPTP(A,B)$ be a bosonic dephasing channel, and let $\rho\in\dop{=}(RA)$ satisfy
    \begin{align}
        \Tr[\hat{n}\rho_A]\leq E<\infty.
    \end{align}
    Then
    \begin{align}\label{eq:div to Ic}
        D\left((\idmap_R\otimes\bd{\mu})[\rho_{RA}]\Vert(\idmap_R\otimes\Delta)[\rho_{RA}]\right)\leq\coh{\rho}{A},
    \end{align}
    where $\Delta\in\CPTP(A,B)$ is the completely dephasing channel. Equality holds whenever $\rho_{RA}$ is pure.
\end{lemma}
\begin{proof}
    Let $\psi_{\widehat{R}RA}=\ketbra{\psi}{\psi}_{\widehat{R}RA}$ be any purification of $\rho_{RA}$, and let
    \begin{align}
        \widetilde{\rho}_{\widehat{R}RB}\defvar(\idmap_{\widehat{R}R}\otimes\bd{\mu})[\psi_{\widehat{R}RA}],\qquad\widehat{\rho}_{\widehat{R}RB}\defvar(\idmap_{\widehat{R}R}\otimes\Delta)[\psi_{\widehat{R}RA}].
    \end{align}
    Setting $\Delta_A,\Delta_B$ to be completely dephasing {channels} on $A,B$ respectively, we get
    \begin{align}\label{eq:CD rho}
        (\idmap_{\widehat{R}R}\otimes\Delta_B)[\widetilde{\rho}_{\widehat{R}RB}]&=(\idmap_{\widehat{R}R}\otimes\Delta_B\circ\bd{\mu})[\psi_{\widehat{R}RA}]\nonumber\\
        &=(\idmap_{\widehat{R}R}\otimes\Delta)[\psi_{\widehat{R}RA}]\nonumber\\
        &=\widehat{\rho}_{\widehat{R}RB},
    \end{align}
    where the second line follows from Eq.~\eqref{eq:BD and CD}.
   Therefore, we have\footnote{The logarithmic steps are justified by applying
self-adjointness of $\Delta_B$ first to bounded spectral truncations
of $-\log\widehat{\rho}$, which are fixed by $\Delta_B$, and then
using monotone convergence. This gives
$\Tr[\widetilde{\rho}(-\log\widehat{\rho})]
=H(\widehat{R}RB)_{\widehat{\rho}}$.
The finiteness of both entropy terms is established below.}
    \begin{align}\label{eq:div to ent}
        D\left((\idmap_R\otimes\bd{\mu})[\rho_{RA}]\Vert(\idmap_R\otimes\Delta)[\rho_{RA}]\right)&\leq D\left((\idmap_{\widehat{R}R}\otimes\bd{\mu})[\psi_{\widehat{R}RA}]\Vert(\idmap_{\widehat{R}R}\otimes\Delta)[\psi_{\widehat{R}RA}]\right)\nonumber\\        &=D\left(\widetilde{\rho}_{\widehat{R}RB}\Vert\widehat{\rho}_{\widehat{R}RB}\right)\nonumber\\
        &=-H(\widehat{R}RB)_{\widetilde{\rho}}-\Tr\left[\widetilde{\rho}_{\widehat{R}RB}\log\widehat{\rho}_{\widehat{R}RB}\right]\nonumber\\
        &=-H(\widehat{R}RB)_{\widetilde{\rho}}-\Tr\left[\widetilde{\rho}_{\widehat{R}RB}\Delta_B[\log\widehat{\rho}_{\widehat{R}RB}]\right]\nonumber\\
        &=-H(\widehat{R}RB)_{\widetilde{\rho}}-\Tr\left[\Delta_B[\widetilde{\rho}_{\widehat{R}RB}]\log\widehat{\rho}_{\widehat{R}RB}\right]\nonumber\\
        &=-H(\widehat{R}RB)_{\widetilde{\rho}}-\Tr\left[\widehat{\rho}_{\widehat{R}RB}\log\widehat{\rho}_{\widehat{R}RB}\right]\nonumber\\
        &=H(\widehat{R}RB)_{\widehat{\rho}}-H(\widehat{R}RB)_{\widetilde{\rho}},
    \end{align}
    where the first line holds due to data processing, the fourth line follows from Eq.~\eqref{eq:CD rho} by noting that $\Delta_B^2[\widetilde{\rho}]=\Delta_B[\widehat{\rho}]$ and $\Delta_B^2[\widetilde{\rho}]=\Delta_B[\widetilde{\rho}]=\widehat{\rho}$, the fifth line holds since $\Delta_B$ is self-adjoint (i.e., $\Delta_B=\Delta_B^\dagger$), and the sixth line follows from Eq.~\eqref{eq:CD rho}. It now suffices to bound the last line of Eq.~\eqref{eq:div to ent}. Observe that since $\psi_{\widehat{R}RA}$ is pure, $\widehat{\rho}_{\widehat{R}RB}$ is classical on $B$, with pure conditional states on $\widehat{R}R$ register. Hence, there {exist} pure states $\{\ket{\phi_i}_{\widehat{R}R}\}_i$ such that
    \begin{align}
        \widehat{\rho}_{\widehat{R}RB}=\sum_i{\rho(i)}\ketbra{\phi_i}{\phi_i}_{\widehat{R}R}\otimes\ketbra{i}{i}_B.
    \end{align}
    Therefore, recalling $\bsym{\rho}_A=\{{\rho(i)}\}_i$ 
    \begin{align}
        H(\widehat{R}RB)_{\widehat{\rho}}=H(\bsym{\rho}_A).
    \end{align}
    Since $\Delta_A$ preserves photon number,
    \begin{align}
        \Tr\left[\hat{n}\Delta_A[\rho_A]\right]=\Tr[\hat{n}\rho_A]\leq E.
    \end{align}
    Applying Lemma~\ref{lemma:finiteness} gives
    \begin{align}
        H(\bsym{\rho}_A)=H(A)_{\Delta_A[{\rho_A}]}\leq g(E)<\infty.
    \end{align}
    Hence $H(\widehat{R}RB)_{\widehat{\rho}}<\infty$. To prove the finiteness of the other term, let $\widetilde{\rho}_{\widehat{R}RB{F}}$ be the pure purification of $\widetilde{\rho}_{\widehat{R}RB}$ obtained by applying a Stinespring dilation of $\bd{\mu}$ to $\psi_{\widehat{R}RA}$. Then we have
    \begin{align}
        H(\widehat{R}RB)_{{\widetilde{\rho}}}&=H({F})_{\widetilde{\rho}}\nonumber\\
        &=H({F})_{\bdc{\mu}[\rho_A]}\nonumber\\
        &\leq g(E)<\infty,
    \end{align}
    where in the first line we invoked the duality of von Neumann entropy, the second line follows from the definition of the complementary channel, and the last line follows from Lemma~\ref{lemma:finiteness}. Both entropy terms in~\eqref{eq:div to ent} are therefore finite.
    Using $\bdc{\mu}\circ\Delta_A=\bdc{\mu}$, we consequently obtain
    \begin{align}\label{eq:IC}
        H(\widehat{R}RB)_{\widehat{\rho}}-H(\widehat{R}RB)_{\widetilde{\rho}}&=H(\bsym{\rho}_A)-H({F})_{\bdc{\mu}\circ\Delta_A[\rho_A]}\nonumber\\
        &=\coh{\rho}{A}
    \end{align}
    Combining~\eqref{eq:IC} with~\eqref{eq:div to ent} {proves} Eq.~\eqref{eq:div to Ic}. If $\rho_{AR}$ is pure, then the inequality in~\eqref{eq:div to ent} is saturated and hence we obtain equality in~\eqref{eq:div to Ic}.
\end{proof}
\subsection{Quantum capacity}
\label{sec:QC}
{The next lemma is a special case of}~\cite[Theorem~2]{P91}. We include a proof for completeness.
\begin{lemma}
    \label{lemma:pinching chain rule}
    Let $\{\Pi_i\}_i$ be a countable family of mutually orthogonal projections satisfying $\sum_i\Pi_i=\id_A$ {in the strong operator topology}. Define the pinching channel $\Pmap\in\CPTP(A,A)$ as
    \begin{align}
        \Pmap[\tau]\defvar\sum_i\Pi_i\tau\Pi_i.
    \end{align}
    Then, for every $\rho\in\dop{=}(A)$ and $\sigma\in\tcop(A)$ satisfying $\Pmap[\sigma]=\sigma$, we have
    \begin{align}
        \label{eq:pinching chain rule}
        D(\rho\Vert\sigma)=D\left(\rho\Vert\Pmap[\rho]\right)+D\left(\Pmap[\rho]\Vert\sigma\right)
    \end{align}
\end{lemma}
\begin{proof}
If $\sigma=0$, both sides are $+\infty$, so assume $\sigma\ne0$.
We first prove the finite-dimensional case. The inclusion
$\supp(\rho)\subseteq\supp(\Pmap[\rho])$ always holds.
If $\supp(\Pmap[\rho])\nsubseteq\supp(\sigma)$, pinching invariance
of $\sigma$ implies that both sides of the claimed identity are
$+\infty$. Otherwise, with logarithms evaluated on the relevant
supports, we have

    \begin{align}
        \Tr\left[\rho\left(\log\Pmap[\rho]-\log\sigma\right)\right]&=\Tr\left[\rho\left(\log\Pmap[\rho]-\log\Pmap[\sigma]\right)\right]\nonumber\\
        &=\Tr\left[\rho\Pmap[\log\Pmap[\rho]-\log\Pmap[\sigma]]\right]\nonumber\\
        &=\Tr\left[\Pmap[\rho]\left(\log\Pmap[\rho]-\log\Pmap[\sigma]\right)\right],
    \end{align}
    where the first line follows from the assumption, and the last line holds since $\Pmap$ is self-adjoint. With this in hand, we have
    \begin{align}\label{eq:finite-dim}
        D(\rho\Vert\sigma)&=D(\rho\Vert\Pmap[\rho])+\Tr\left[\rho\left(\log\Pmap[\rho]-\log\sigma\right)\right]\nonumber\\
        &=D(\rho\Vert\Pmap[\rho])+\Tr\left[\Pmap[\rho]\left(\log\Pmap[\rho]-\log\sigma\right)\right]\nonumber\\
        &=D(\rho\Vert\Pmap[\rho])+D(\Pmap[\rho]\Vert\sigma).
    \end{align}
    Having established the finite-dimensional case, we now move to the infinite-dimensional case. Set $\mathcal{H}_A^i\defvar\Pi_i\mathcal{H}_A$, then $\mathcal{H}_A=\bigoplus_i\mathcal{H}_A^i$. For each $i$, choose an orthonormal basis $\{\ket{e}_{i,m}\}_m$ of $\mathcal{H}_A^i$, and observe that their union is a countable orthonormal basis of $\mathcal{H}_A$. Enumerate this union as $\{\ket{e_j}\}_j$ and define the finite-rank orthogonal projections
    \begin{align}
        Q^k\defvar\sum_{j=1}^k\ketbra{e_{j}}{e_{j}}.
    \end{align}
    Note that since each $\ket{e_j}$ belongs to one of the $\mathcal{H}_A^i$, $[\ketbra{e_j}{e_j},\Pi_i]=0$, and hence $[Q^k,\Pi_i]=0$ for every ${i,k}$. Moreover, we have that $Q^k\leq Q^{k+1}$, and $Q^k$ converges strongly to $\id_A$ as $k\to\infty$. Let $\Phi^k\in\CPTP(A,A_k)$ be such that
    \begin{align}\label{eq:Phik}
        \Phi^k[\tau]\defvar Q^k\tau Q^k\oplus\Tr\left[(\id_A-Q^k)\tau\right]\ketbra{\perp}{\perp}.
    \end{align}
    Here $\mathcal{H}_{A_k}\defvar Q^k\mathcal{H}_A\oplus\operatorname{span}\{\ket{\perp}\}$, where $\ket{\perp}$ is a flag state. Define a {pinching} channel $\Pmap^k\in\CPTP(A_k,A_k)$ by
    \begin{align}
        \label{eq:Pk}
        \Pmap^k[\tau]\defvar\sum_j\Pi^k_{j}\tau\Pi^k_{j}+\Pi_{\perp}\tau\Pi_{\perp},
    \end{align}
    where $\Pi^k_{j}=\Pi_jQ^k$ and $\Pi_{\perp}=\ketbra{\perp}{\perp}$. Then for any ${\tau\in\tcop(A)}$
    \begin{align}
        \Phi^k\circ\Pmap[\tau]&=\sum_jQ^k\Pi_j\tau\Pi_jQ^k\oplus\Tr\left[(\id_A-Q^k)\Pmap[\tau]\right]\ketbra{\perp}{\perp}\nonumber\\
        &=\sum_j(Q^k\Pi_jQ^k)(Q^k\tau Q^k)(Q^k\Pi_jQ^k)\oplus\Tr\left[(\id_A-Q^k)\tau\right]\ketbra{\perp}{\perp}\nonumber\\
        &=\Pmap^k\circ\Phi^k[\tau],
    \end{align}
    where in the second line we used $\left(Q^k\right)^2=Q^k$ and $\sum_i\Pi_i=\id_A$. From the above relation, we can conclude that
    \begin{align}\label{eq:finit_cond}
        \Pmap^k\left[\Phi^k[\sigma]\right]=\Phi^k\circ\Pmap[\sigma]=\Phi^k[\sigma],\ \textit{and}\qquad\Pmap^k\circ\Phi^k[\rho]=\Phi^k\circ\Pmap[\rho] 
    \end{align}
    and therefore,
    \begin{align}\label{eq:truncated equality}
        D\left(\Phi^k[\rho]\Vert\Phi^k[\sigma]\right)&=D\left(\Phi^k[\rho]\Vert\Pmap^k\circ\Phi^k[\rho]\right)+D\left(\Pmap^k\circ\Phi^k[\rho]\Vert\Phi^k[\sigma]\right)\nonumber\\
        &=D\left(\Phi^k[\rho]\Vert\Phi^k\circ\Pmap[\rho]\right)+D\left(\Phi^k\circ\Pmap[\rho]\Vert\Phi^k[\sigma]\right),
    \end{align}
    where in the first line we applied~\eqref{eq:finite-dim}, and second line follows by substituting from~\eqref{eq:finit_cond}.
    It now suffices to show that for any $\tau\in\dop{=}(A)$ and $\eta\in\tcop(A)$, $\lim_{k\to\infty}D\left(\Phi^k[\tau]\Vert\Phi^k[\eta]\right)=D(\tau\Vert\eta)$. To do so, let $\epsilon_k\defvar\Tr[(\id_A-Q^k)\tau]$, then from gentle operator lemma~\cite[Lemma~9.4.2]{W17}, we have
    \begin{align}
        \norm{\Phi^k[\tau]-\left(\tau\oplus 0\right)}_1&=\norm{\left(Q^k\tau Q^k-\tau\right)\oplus\left(\epsilon_k-0\right)}_1\nonumber\\
        &=\norm{Q^k\tau Q^k-\tau}_1+\epsilon_k\nonumber\\
        &\leq 2\sqrt{\epsilon_k}+\epsilon_k.
    \end{align}
    Therefore, since $Q^k$ converges to $\id_A$ strongly, we get $\lim_{k\to\infty}\epsilon_k=0$, and hence
    \begin{align}\label{eq:convergence phi}
        \lim_{k\to\infty}\norm{\Phi^k[\tau]-\left(\tau\oplus 0\right)}_1=0.
    \end{align}
    {For $\eta\ne0$, applying the same argument to
    $\eta/\Tr\eta$ gives $\Phi^k[\eta]\to\eta\oplus0$ in trace norm;
    for $\eta=0$ this is immediate. With these two convergences, we obtain}
    \begin{align}
        D(\tau\Vert\eta)&=D(\tau\oplus0\Vert\eta\oplus0)\nonumber\\
        &\leq \liminf_{k\to\infty}{D\left(\Phi^k[\tau]\Vert\Phi^k[\eta]\right)}\nonumber\\
        &\leq \limsup_{k\to\infty}{D\left(\Phi^k[\tau]\Vert\Phi^k[\eta]\right)}\nonumber\\
        &\leq D(\tau\Vert\eta),
    \end{align}
    where the second line follows from joint lower semicontinuity of relative entropy~\cite[Lemma~4]{LS23} together with the convergence in~\eqref{eq:convergence phi}, and the last line is an application of data processing. {The first summand in~\eqref{eq:truncated equality} is nonnegative,
and the second is bounded below by $-\log\Tr\sigma$, uniformly in $k$.
We may therefore pass to the limit in that identity, also when a term
is infinite, to obtain the claim.}
\end{proof}
\begin{theorem}\label{thrm:EC chain rule}
    (Energy-constrained chain rule for bosonic dephasing channels)\par\noindent
    Let $\bd{\mu}\in\CPTP(A,B)$ be a bosonic dephasing channel, and let $\rho\in\dop{=}(RA)$ satisfy
    \begin{align}
        \Tr[\hat{n}\rho_A]\leq E<\infty.
    \end{align}
    Then, for any $\sigma\in\tcop(RA)$,
    \begin{align}\label{eq:EC chain rule}
        D\left((\idmap_R\otimes\bd{\mu})[\rho_{RA}]\Vert(\idmap_R\otimes\Delta)[\sigma_{RA}]\right)\leq D(\rho_{RA}\Vert\sigma_{RA})+\coh{\rho}{A},
    \end{align}
    where $\Delta\in\CPTP(A,B)$ is the completely dephasing channel.
\end{theorem}
\begin{proof}
Let $\Pi_i=\id_R\otimes\pure{i}_B$, with $\{\ket{i}_B\}_i$
the Fock basis. The associated pinching on $RB$ is
$\Pmap=\idmap_R\otimes\Delta_B$.
By Lemma~\ref{lemma:BD and CD} and $\Delta_B\circ\Delta=\Delta$,
Lemma~\ref{lemma:pinching chain rule} gives

    \begin{align}
        &D\left((\idmap_R\otimes\bd{\mu})[\rho_{RA}]\Vert(\idmap_R\otimes\Delta)[\sigma_{RA}]\right)\nonumber\\
        &\qquad\qquad=D\left((\idmap_R\otimes\bd{\mu})[\rho_{RA}]\Vert(\idmap_R\otimes\Delta)[\rho_{RA}]\right)+D\left((\idmap_R\otimes\Delta)[\rho_{RA}]\Vert(\idmap_R\otimes\Delta)[\sigma_{RA}]\right)\nonumber\\
        &\qquad\qquad\leq\coh{\rho}{A}+D(\rho_{RA}\Vert\sigma_{RA}),
    \end{align}
    where the inequality follows from Lemma~\ref{lemma:div to Ic} and data processing for the first and second terms, respectively.
\end{proof}
\begin{theorem}[Rains relative entropy chain rule for bosonic dephasing channels]
\label{lem:rains-increment}
Let $\mathcal{N}_\mu\in\CPTP(A,B)$ be a bosonic
dephasing channel, and let $\rho\in\mathrm D_{=}(A'AB')$ satisfy
\begin{equation}
    \Tr[\hat{n} \rho_A]\le E<\infty.
\end{equation}
Then,
\begin{equation}
    R(A':BB')_{\mathcal{N}_\mu[\rho_{A'AB'}]}
    \leq R(A'A:B')_\rho+\coh{\rho}{A}.
\end{equation}
\end{theorem}

\begin{proof}
Let $\sigma_{A'AB'}\in\PPT'(A'A:B')$, and set
\begin{equation}
    \tau_{A'BB'}\defvar\Delta[\sigma_{A'AB'}],
\end{equation}
where $\Delta\in\CPTP(A,B)$ is the completely
dephasing channel. We first show that
$\tau\in\PPT'(A':BB')$. Taking the partial transpose
{on $BB'$, with the Fock basis on $B$}, we have
\begin{equation}
    \tau_{A'BB'}^{T_{BB'}}
    =\Delta\left[\sigma_{A'AB'}^{T_{B'}}\right],
\end{equation}
since the Fock projectors prepared by $\Delta$ are invariant
under transposition. {Feasibility of $\sigma$ makes $\sigma^{T_{B'}}$ trace class;
applying $\Delta$ therefore gives a trace-class operator.} Consequently,
\begin{align}
   \norm{\tau_{A'BB'}^{T_{BB'}}}_1
    &=
    \norm{\Delta\left[\sigma_{A'AB'}^{T_{B'}}\right]}_1
    \nonumber\\
    &\leq
    \norm{\sigma_{A'AB'}^{T_{B'}}}_1
    \le 1,
\end{align}
where the first inequality follows from trace-norm contractivity
on Hermitian operators. Since $\tau$ is also positive and trace
class, it is feasible in the output Rains optimization.

Applying Theorem~\ref{thrm:EC chain rule} with reference system $A'B'$, we therefore obtain
\begin{align}
    R(A':BB')_{\mathcal{N}_\mu[\rho_{A'AB'}]}
    &\leq D\left(\mathcal{N}_\mu[\rho_{A'AB'}]\Vert\Delta[\sigma_{A'AB'}]\right)
    \nonumber\\
    &\le D(\rho_{A'AB'}\Vert\sigma_{A'AB'})+\coh{\rho}{A}.
\end{align}
Taking the infimum over
$\sigma\in\PPT'(A'A:B')$ proves the claim.
\end{proof}
\begin{theorem}[Energy-constrained quantum capacities of bosonic dephasing channels]
\label{thrm:quantum-capacities}
Let $\bd{\mu}\in\CPTP(A,B)$ be a bosonic dephasing channel, and let
$0\leq E<\infty$. Then,
\begin{align}
 \Qt{\bd{\mu},E}=Q(\bd{\mu},E).
 \label{eq:quantum-capacities}
\end{align}
Here the assisted protocols satisfy
\begin{align}
 \frac1n\sum_{i=1}^n\Tr[\hat n\rho_{A_i}^{(i)}]\leq E,
\end{align}
where $\rho_{A_i}^{(i)}$ is the unconditional input marginal before
the $i$-th channel use.
\end{theorem}

\begin{proof}
For achievability, use the equivalent average-input entanglement-transmission
formulation recalled in Definition~\ref{def:unassisted-capacities}.
Alice locally prepares a maximally entangled state, applies the encoder
of an unassisted code to one half, and sends the encoded systems through
the channel. Bob then applies the decoder. This is a protocol of
Definition~\ref{def:assisted-quantum}, and its average input energy is
exactly the quantity in~\eqref{eq:block-energy}. Thus every unassisted
achievable rate is also achievable with two-way LOCC assistance, giving
\begin{align}\label{eq:lower2w}
    \Qt{\bd{\mu},E}\geq Q\left(\bd{\mu},E\right).
\end{align}
To prove the converse, {consider} an $(n,K,\hat n,E,\varepsilon)$ quantum communication protocol
as in Definition~\ref{def:assisted-quantum}. Let $\rho^{(i)}_{A_i'A_iB_i'}$ and
$\sigma^{(i)}_{A_i'B_iB_i'}$ be the states immediately before and after
the $i$-th channel use, respectively. Thus
\begin{align}
  \sigma^{(i)}_{A_i'B_iB_i'}
   &= ({\idmap_{A_i'}}\otimes\mathcal{N}_{\mu,A_i\to B_i}\otimes{\idmap_{B_i'}})
       [\rho^{(i)}_{A_i'A_iB_i'}], && i=1,\ldots,n,\\
  \rho^{(i+1)}_{A_{i+1}'A_{i+1}B_{i+1}'}
   &=\mathcal{L}^{(i+1)}[\sigma^{(i)}_{A_i'B_iB_i'}], && i=1,\ldots,n-1,\\
  \omega_{K_AK_B}&=\mathcal{L}^{(n+1)}[\sigma^{(n)}_{A_n'B_nB_n'}].
\end{align}
Here $\mathcal{L}^{(i+1)}\in\mathrm{CPTP}(A_i'B_iB_i',A_{i+1}'A_{i+1}B_{i+1}')$
and $\mathcal{L}^{(n+1)}\in\mathrm{CPTP}(A_n'B_nB_n',K_AK_B)$ are the intervening
and final LOCC maps, respectively.
Write
$\rho^{(i)}_{A_i}=\Tr_{A_i'B_i'}[\rho^{(i)}]$, and let
$\boldsymbol{\rho}^{(i)}_{A_i}$ denote its Fock-basis distribution.

LOCC maps preserve $\PPT'$. Thus, for every
$\zeta\in\PPT'(A_i':B_iB_i')$, data processing gives
\begin{align}
  R(A_{i+1}'A_{i+1}:B_{i+1}')_{\rho^{(i+1)}}
   \le D\bigl(\mathcal{L}^{(i+1)}[\sigma^{(i)}]\big\|
                   \mathcal{L}^{(i+1)}[\zeta]\bigr)
   \le D(\sigma^{(i)}\|\zeta).
\end{align}
Taking the infimum over $\zeta$ proves that these maps cannot increase
Rains relative entropy; the same argument applies to the final map.
Iterating Theorem~\ref{lem:rains-increment} and using the separability of $\rho^{(1)}$ therefore gives
\begin{align}
  R(K_A:K_B)_\omega\le\sum_{i=1}^n
                              I_c(\boldsymbol{\rho}^{(i)}_{A_i}).
\end{align}
Identifying each $A_i$ with $A$, set
$\tau_A:=n^{-1}\sum_{i=1}^n\Delta_A[\rho^{(i)}_{A_i}]$,
and let $\boldsymbol{\tau}_A$ denote its Fock-basis distribution.
The energy constraint implies $\Tr[\hat n\tau_A]\le E$.
Each input has finite mean photon number, since the nonnegative terms
in~\eqref{eq:adaptive-energy} have a finite sum. Concavity of coherent
information~\cite[Proposition~1]{WQ18}, justified by
Lemma~\ref{lemma:finiteness}, and
Corollary~\ref{cor:bdc-unassisted-capacities} give
\begin{align}
  R(K_A:K_B)_\omega
   \le\sum_{i=1}^n I_c(\boldsymbol{\rho}^{(i)}_{A_i})
   \le n I_c(\boldsymbol{\tau}_A)
   \le nQ(\mathcal{N}_\mu,E).
\end{align}

For every $\zeta\in\PPT'(K_A:K_B)$, H\"older's inequality gives
\begin{align}
  \Tr[\Phi^{(K)}\zeta]
   \le\bigl\|(\Phi^{(K)})^{T_{K_B}}\bigr\|_\infty
      \bigl\|\zeta^{T_{K_B}}\bigr\|_1
   \le\frac1K,
  \qquad \Tr\zeta\le1.
\end{align}
Since $\Tr[\Phi^{(K)}\omega]\ge1-\varepsilon$, applying data processing to the measurement
$\{\Phi^{(K)},{\id_{K_AK_B}}-\Phi^{(K)}\}$ yields
\begin{align}
  D(\omega\|\zeta)\ge \Tr[\Phi^{(K)}\omega]\log K-{\binh}\left(\Tr[\Phi^{(K)}\omega]\right)
                    \ge(1-\varepsilon)\log K-1,
\end{align}
where ${\binh}$ is the binary entropy. Taking the infimum over $\zeta$ and
using the preceding upper bound, we obtain
\begin{align}
  (1-\varepsilon)\log K-1
   \le R(K_A:K_B)_\omega
   \le nQ(\mathcal{N}_\mu,E).
\end{align}
Dividing by $n$ and taking $n\to\infty$ and ${\varepsilon\to0^+}$
in Definition~\ref{def:assisted-quantum} proves ${\Qt{\mathcal{N}_\mu,E}}\le Q(\mathcal{N}_\mu,E)$. Combining this with~\eqref{eq:lower2w} proves the claim.
\end{proof}

\subsection[Private and secret-key capacities]{Private and secret-{key} capacities}
\label{sec:private-capacities}

\begin{theorem}[Relative entropy of entanglement chain rule for bosonic dephasing channels]
\label{thrm:REE chain rule}
Let $\bd{\mu}\in\CPTP(A,B)$ be a bosonic dephasing channel, and let
$\rho\in\dop{=}(A'AB')$ satisfy
\begin{align}
    \Tr[\hat n\rho_A]\leq E<\infty.
\end{align}
Then,
\begin{align}
    E_R(A':BB')_{\bd{\mu}[\rho_{A'AB'}]}
    \leq E_R(A'A:B')_\rho+\coh{\rho}{A}.
    \label{eq:REE chain rule}
\end{align}
\end{theorem}

\begin{proof}
Let $\sigma\in\SEP(A'A:B')$, and set
$\tau_{A'BB'}\defvar\Delta[\sigma_{A'AB'}]$, where
$\Delta\in\CPTP(A,B)$ is the completely dephasing channel.
We first show that $\tau\in\SEP(A':BB')$. If
$\sigma_{A'AB'}=\sigma_{A'A}\otimes\sigma_{B'}$, then
\begin{align}
    \tau_{A'BB'}
    =\sum_{k=0}^{\infty}
      \bra{k}_A\sigma_{A'A}\ket{k}_A
      \otimes\pure{k}_B\otimes\sigma_{B'}.
\end{align}
After normalizing the nonzero summands, this is a countable convex
combination of product states across $A':BB'$.
By linearity, the same holds for finite convex combinations of
product states. Every separable state is a trace-norm limit of such
combinations, so continuity of $\Delta$ and closedness of the
separable set give $\Delta[\sigma]\in\SEP(A':BB')$ for every
$\sigma\in\SEP(A'A:B')$.
Applying Theorem~\ref{thrm:EC chain rule} with reference system $A'B'$ gives
\begin{align}
    E_R(A':BB')_{\bd{\mu}[\rho_{A'AB'}]}
    &\leq D\bigl(\bd{\mu}[\rho_{A'AB'}]\big\|\Delta[\sigma_{A'AB'}]\bigr)
      \nonumber\\
    &\leq D(\rho_{A'AB'}\|\sigma_{A'AB'})+\coh{\rho}{A}.
\end{align}
Taking the infimum over $\sigma\in\SEP(A'A:B')$ proves the claim.
\end{proof}

\begin{theorem}[Energy-constrained private capacities of bosonic dephasing channels]
\label{thrm:private-capacities}
Let $\bd{\mu}\in\CPTP(A,B)$ be a bosonic dephasing channel, and let
$0\leq E<\infty$. Then,
\begin{align}
    \Pt{\bd{\mu},E}=P(\bd{\mu},E)=Q(\bd{\mu},E).
    \label{eq:private-capacities}
\end{align}
\end{theorem}

\begin{proof}
A maximally entangled state is a private state with trivial shields.
Thus, Eq.~\eqref{eq:lower2w} and
Corollary~\ref{cor:bdc-unassisted-capacities} give
\begin{align}
    \Pt{\bd{\mu},E}\geq\Qt{\bd{\mu},E}
    \geq Q(\bd{\mu},E)=P(\bd{\mu},E).
    \label{eq:private-achievability}
\end{align}

For the converse, consider an $(n,K,\hat n,E,\eps)$ protocol from
Definition~\ref{def:assisted-private}, with final state
$\omega_{K_AS_AK_BS_B}$ satisfying $F(\omega,\gamma)\geq1-\eps$
for a private state $\gamma$ of key dimension $K$.

Use the states $\rho^{(i)}$ and $\sigma^{(i)}$ from
Eqs.~\eqref{eq:adaptive-channel-use}--\eqref{eq:adaptive-locc}.
The intervening and final LOCC maps preserve separability, so data
processing implies that they cannot increase $E_R$.
Iterating Theorem~\ref{thrm:REE chain rule} from the separable initial
state and applying the same averaging argument as in the quantum
converse gives
\begin{align}
    E_R(K_AS_A:K_BS_B)_\omega
    &\leq\sum_{i=1}^n I_c(\bsym{\rho}^{(i)}_{A_i})
    \leq nQ(\bd{\mu},E).
    \label{eq:private-ER-upper}
\end{align}
\begin{samepage}
Let $U$ be the twisting unitary of $\gamma$ in
Definition~\ref{def:private-state}, and define the privacy test~\cite{WTB17}
\begin{align}
    \Pi_\gamma\defvar
    U\bigl(\Phi^{(K)}_{K_AK_B}\otimes\id_{S_AS_B}\bigr)U^\dagger.
    \label{eq:private-test}
\end{align}
Then $\Tr[\Pi_\gamma\gamma]=1$.
For a product state
$\zeta=\zeta_{K_AS_A}\otimes\zeta_{K_BS_B}$, write
$Z_{ij}\defvar\bra{ii}\zeta\ket{jj}$, where the matrix elements
are taken on the key registers $K_AK_B$.
H\"older's inequality gives
$\norm{Z_{ij}}_1\leq\sqrt{\Tr[Z_{ii}]\Tr[Z_{jj}]}$.
Expanding $\Pi_\gamma$ and using this bound, we obtain
\begin{align}
    \Tr[\Pi_\gamma\zeta]
    &\leq\frac1K\sum_{i,j=1}^K\norm{Z_{ij}}_1\nonumber\\
    &\leq\frac1K\left(\sum_{i=1}^K\sqrt{\Tr[Z_{ii}]}\right)^2
      \leq\frac1K.
\end{align}
\end{samepage}
The last inequality follows from Cauchy--Schwarz and
$\Tr[Z_{ii}]=\bra{i}\zeta_{K_A}\ket{i}\bra{i}\zeta_{K_B}\ket{i}$.
By convexity and trace-norm continuity,
$\Tr[\Pi_\gamma\zeta]\leq1/K$ for every
$\zeta\in\SEP(K_AS_A:K_BS_B)$. This argument also applies to
infinite-dimensional shields.

Set $p\defvar\Tr[\Pi_\gamma\omega]$.
Then $p\geq F(\omega,\gamma)\geq1-\eps$ by monotonicity of fidelity
under the measurement $\{\Pi_\gamma,\id-\Pi_\gamma\}$.
For any $\zeta\in\SEP(K_AS_A:K_BS_B)$, put
$q\defvar\Tr[\Pi_\gamma\zeta]$.
Since $q\leq1/K$, data processing for the same measurement gives
\begin{align}
    D(\omega\|\zeta)
    &\geq p\log\frac{p}{q}+(1-p)\log\frac{1-p}{1-q}\nonumber\\
    &\geq p\log K-\binh(p)
     \geq(1-\eps)\log K-1,
\end{align}
where $\binh$ is the binary entropy.
Taking the infimum over $\zeta$ and using
Eq.~\eqref{eq:private-ER-upper}, we obtain
\begin{align}
    (1-\eps)\log K-1
    \leq E_R(K_AS_A:K_BS_B)_\omega
    \leq nQ(\bd{\mu},E).
\end{align}
Dividing by $n$ and taking $n\to\infty$ and $\eps\to0^+$
in Definition~\ref{def:assisted-private} proves
$\Pt{\bd{\mu},E}\leq Q(\bd{\mu},E)$.
Together with Eq.~\eqref{eq:private-achievability}, this proves the claim.
\end{proof}
The following corollary shows that the bounds in {Theorems}~\ref{lem:rains-increment} and~\ref{thrm:REE chain rule} are attained.
\begin{corollary}[Tightness of the entanglement chain rules]
\label{cor:entanglement-chain-rule-tightness}
Let $\bd{\mu}\in\CPTP(A,B)$ be a bosonic dephasing channel, and let
$\rho_A\in\dop{=}(A)$ be diagonal in the Fock basis and satisfy
\begin{align}
    \Tr[\hat n\rho_A]\leq E<\infty.
\end{align}
Let $\ket{\psi}_{RA}$ be the canonical purification of $\rho_A$ in the
Fock basis,
\begin{align}
    \ket{\psi}_{RA}
    \defvar\sum_{k=0}^{\infty}\sqrt{\rho(k)}\,
               \ket{k}_R\otimes\ket{k}_A,
\end{align}
where $\psi_{RA}\defvar\pure{\psi}_{RA}$. Then,
\begin{align}
    R(R:B)_{(\idmap_R\otimes\bd{\mu})[\psi_{RA}]}
    =E_R(R:B)_{(\idmap_R\otimes\bd{\mu})[\psi_{RA}]}
    =\coh{\rho}{A}.
    \label{eq:entanglement-chain-rule-tightness}
\end{align}
\end{corollary}

\begin{proof}
The output $(\idmap_R\otimes\bd{\mu})[\psi_{RA}]$ is maximally
correlated. Its $R$-marginal has entropy $H(A)_\rho$, which is finite
by the energy constraint.
By Lemma~\ref{lemma:BD and CD}, its dephased state is
\begin{align}
    (\idmap_R\otimes\Delta)[\psi_{RA}]
    =\sum_{k=0}^{\infty}\rho(k)\,\pure{k}_R\otimes\pure{k}_B.
\end{align}
Lemma~\ref{lemma:Rains MM} and the pure-input equality in
Lemma~\ref{lemma:div to Ic} therefore give
\begin{align}
    R(R:B)_{(\idmap_R\otimes\bd{\mu})[\psi_{RA}]}
    &=D\bigl((\idmap_R\otimes\bd{\mu})[\psi_{RA}]
        \big\|(\idmap_R\otimes\Delta)[\psi_{RA}]\bigr)\nonumber\\
    &=\coh{\rho}{A}.
\end{align}
Every separable state is Rains-feasible: this holds for finite
mixtures of product states and extends to their trace-norm closure
by Hilbert--Schmidt continuity of partial transposition and lower
semicontinuity of the trace norm.
Since $(\idmap_R\otimes\Delta)[\psi_{RA}]$ is separable, we have

\begin{align}
    R(R:B)_{(\idmap_R\otimes\bd{\mu})[\psi_{RA}]}
    \leq E_R(R:B)_{(\idmap_R\otimes\bd{\mu})[\psi_{RA}]}
    \leq\coh{\rho}{A},
\end{align}
which proves Eq.~\eqref{eq:entanglement-chain-rule-tightness}.

Taking $A'=R$ and $B'$ trivial makes both input entanglement measures
zero. Thus the bounds in Theorems~\ref{lem:rains-increment}
and~\ref{thrm:REE chain rule} are attained for every such input marginal
$\rho_A$.
\end{proof}
\section{Conclusion}
\label{sec:conclusion}

In this work, we showed that the quantum and private capacities of
single-mode bosonic dephasing channels are unchanged by two-way
classical assistance under an average photon-number constraint.
For every finite energy constraint, all four capacities are given by the
same optimization of $\coh{\rho}{A}$ over photon-number distributions
satisfying the constraint. This equality holds for adaptive protocols
in which the input energy may vary between channel uses.

Our proof is based on a relative entropy comparison between the
bosonic dephasing channel and the completely dephasing channel.
Combining this comparison with a pinching identity, we derived chain
rules that bound the increase in the Rains relative entropy and the
relative entropy of entanglement across a channel use by
$\coh{\rho}{A}$. Applying these bounds successively, together with monotonicity
under LOCC and concavity of the coherent information, shows that
interactive classical assistance does not increase either capacity
under the average energy constraint.
Furthermore, both chain rules are attained by the canonical purifications of
finite-energy Fock-diagonal inputs. The analysis includes countably supported
photon-number distributions and allows infinite-dimensional local
memories and private-state shields.

Our characterization leaves an optimization over photon-number
distributions. Identifying optimal or near-optimal distributions
for specific phase-noise models would clarify how the capacity
depends on the energy constraint.

Another direction is to study bosonic loss-dephasing channels,
which combine photon loss and dephasing and have been investigated
in~\cite{LXJR22,LW23,MSGL24,WHLX26,HLSW26}.
It would be useful to determine whether the methods developed here can sharpen existing bounds on their
energy-constrained two-way assisted quantum and private capacities. One possible approach is to combine our methods with those developed for the pure-loss channel~\cite{Pir26,FMLG26}.

\section{Acknowledgment}
AA conducted research at the Institute for Quantum
Computing, at the University of Waterloo, which is supported by Innovation, Science, and Economic Development
Canada. Support was also provided by NSERC under the
Discovery Grants Program, Grant No. 341495.
\section{Statement on the use of artificial intelligence}
ChatGPT Astra was used extensively as a research and writing assistant. The authors formulated and directed the research, while ChatGPT Astra provided substantial assistance in developing the main solution, completing proofs, and identifying mathematical and expository errors. The majority of the manuscript was written by the authors, with ChatGPT Astra used iteratively for editing and refinement. All AI-assisted content was reviewed and verified by the authors, who take full responsibility for the final work.
\printbibliography
\appendix
\section{Rains relative entropy of maximally correlated states}
\label{app:maximally-correlated}
\label{app: rains}
\begin{lemma}\label{lemma:Rains MM}
    (Rains relative entropy for countable maximally correlated states)
    Let $\rho_{AB}\in\dop{=}(AB)$ be a maximally correlated state satisfying $H(A)_\rho<\infty$. Set
    \begin{align}
        \hat{\rho}_{AB}\defvar\Delta_B[\rho_{AB}],
    \end{align}
    where $\Delta_B\in\CPTP(B,B)$ is the completely dephasing channel {in any orthonormal basis extending the paired family on $B$}. Then,
    \begin{align}
        \label{eq:Rains MM}
        R(A:B)_\rho=D(\rho_{AB}\Vert\hat{\rho}_{AB})=H(A)_\rho-H(AB)_\rho.
    \end{align}
\end{lemma}
\begin{proof}
    We begin by showing the second equality. Let $\rho_{AB}={\sum_{i,j}}c_{ij}\ketbra{ii}{jj}$, {where the sums run over finite or countable paired orthonormal families. Then}
    \begin{align}\label{eq:hatrho}
        \hat{\rho}_{AB}=\Delta_B[\rho_{AB}]={\sum_i}c_{ii}\ketbra{ii}{ii}.
    \end{align}
    Noting that $\Delta_B$ is unital, we have
    \begin{align}
        H(AB)_\rho\leq H(AB)_{\hat{\rho}}=H(A)_\rho<\infty,
    \end{align}
    {where the first inequality follows from entropy monotonicity under unital channels~\cite[Theorems~2.2 and~3.3]{LB13}.} Consequently, by {a calculation similar to that in}~\eqref{eq:div to ent}, we obtain
    \begin{align}
        D(\rho_{AB}\Vert\hat{\rho}_{AB})&=H(AB)_{\hat{\rho}}-{H(AB)_\rho}\nonumber\\
        &=H(A)_\rho-H(AB)_\rho.
    \end{align}
    We now move to prove the first equality in~\eqref{eq:Rains MM}. First, note that since $\hat{\rho}^{T_B}_{AB}=\hat{\rho}_{AB}$ and $\Tr[\hat{\rho}_{AB}]=1$, $\hat{\rho}_{AB}\in\PPT'(A:B)$. Hence
    \begin{align}\label{eq:Rupper}
        R(A:B)_\rho\leq D(\rho_{AB}\Vert\hat{\rho}_{AB}).
    \end{align}
    For the reverse inequality, we follow the approach used in the proof of Lemma~\ref{lemma:pinching chain rule}, by extending the finite-dimensional result to the infinite-dimensional setting by a limiting argument. For every finite-dimensional maximally correlated state $\eta_{AB}$,
    \begin{align}\label{eq:finiteR}
        R(A:B)_\eta\geq D(\eta_{AB}\Vert\Delta_B[\eta_{AB}]).
    \end{align}
    The result follows from the optimality of the dephased state in the Rains optimization~\cite[Sec.~V, discussion following Eq.~(50)]{FWTD19}. We extend this bound to the countable setting by a limiting argument. Define
    \begin{align}
        Q^k_A\defvar{\sum_{i<k}}\ketbra{i}{i}_A\ , \quad Q^k_B\defvar{\sum_{i<k}}\ketbra{i}{i}_B{.}
    \end{align}
    {These projections exhaust the paired families.} Let $\Phi_A^k\in\CPTP(A,A_k)$ and $\Phi_B^k\in\CPTP(B,B_k)$, be defined as
    \begin{align}
        \label{eq:phiA,phiB}
        &{\Phi_A^k}[\tau]\defvar Q_A^k\tau Q_A^k\oplus\Tr\left[(\id_A-Q_A^k)\tau\right]\ketbra{\perp}{\perp},\qquad{\Phi_B^k}[\tau]\defvar Q_B^k\tau Q_B^k\oplus\Tr\left[(\id_B-Q_B^k)\tau\right]\ketbra{\perp}{\perp}.
    \end{align}
    Here $\mathcal{H}_{A_k}\defvar Q_A^k\mathcal{H}_A\oplus\text{span}\{\ket{\perp}\}$ and $\mathcal{H}_{B_k}\defvar Q_B^k\mathcal{H}_B\oplus\text{span}\{\ket{\perp}\}$. Direct calculation shows 
    \begin{align}
        &\left(\Phi^k_A\otimes\Phi_B^k\right)[\rho_{AB}]=\sum_{i,j<k}c_{ij}\ketbra{ii}{jj}\oplus\Tr\left[(\id_{AB}-Q_A^k\otimes Q_B^k)\rho_{AB}\right]\ketbra{\perp\perp}{\perp\perp},
    \end{align}
     Observe that the above output state is a maximally correlated state, and furthermore,
     \begin{align}\label{eq:dephasek}
         \left(\Phi^k_A\otimes\Phi_B^k\right)\circ\Delta_B[\rho_{AB}]=\Delta_{B_k}\left[\left(\Phi^k_A\otimes\Phi_B^k\right)[\rho_{AB}]\right],
     \end{align}
     where $\Delta_{B_k}$ is a completely dephasing channel on $B_k$. The channels $\Phi_A^k$ and $\Phi_B^k$ commute with {transposition in bases extending the paired families, including the flags at the output}. Thus, for every $\sigma_{AB}\in\PPT'(A:B)$,
     \begin{align}
         \norm{\left(\left(\Phi^k_A\otimes\Phi_B^k\right)[\sigma_{AB}]\right)^{T_{B_k}}}_1&=\norm{\left(\Phi^k_A\otimes\Phi_B^k\right)\left[\sigma_{AB}^{T_B}\right]}_1\nonumber\\
         &\leq\norm{\sigma_{AB}^{T_B}}_1\leq 1,
     \end{align}
     where the first inequality follows from the contractivity of the trace norm. Therefore, we have
     \begin{align}
         D(\rho_{AB}\Vert\sigma_{AB})&\geq D\left(\left(\Phi^k_A\otimes\Phi_B^k\right)[\rho_{AB}]\Vert\left(\Phi^k_A\otimes\Phi_B^k\right)[\sigma_{AB}]\right)\nonumber\\
         &\geq R(A_k:B_k)_{\left(\Phi^k_A\otimes\Phi_B^k\right)[\rho_{AB}]},
     \end{align}
     where the first line follows from data processing. Taking infimum over $\sigma_{AB}\in\PPT'(A:B)$ and using Eqs.~\eqref{eq:finiteR} and~\eqref{eq:dephasek} we obtain
     \begin{align}\label{eq:Rlowerk}
         R(A:B)_\rho\geq R(A_k:B_k)_{\left(\Phi^k_A\otimes\Phi_B^k\right)[\rho_{AB}]}\geq D\left(\left(\Phi^k_A\otimes\Phi_B^k\right)[\rho_{AB}]\Vert\left(\Phi^k_A\otimes\Phi_B^k\right)[\hat{\rho}_{AB}]\right).
     \end{align}
     {We now pass to the limit.} Let $\epsilon_k\defvar\Tr\left[(\id_{AB}-Q_A^k\otimes Q_B^k)\rho_{AB}\right]$. {We compare states on the original local spaces with the flag vectors adjoined, using zero extensions.} Then, using the gentle operator lemma~\cite[Lemma~9.4.2]{W17}, we get
     \begin{align}
         \norm{\left(\Phi^k_A\otimes\Phi_B^k\right)[\rho_{AB}]-(\rho_{AB}\oplus 0)}_1&=\norm{\left(Q_A^k\otimes Q_B^k\right)\rho_{AB}\left(Q_A^k\otimes Q_B^k\right)-\rho_{AB}}_1+\epsilon_k\nonumber\\
         &\leq 2\sqrt{\epsilon_k}+\epsilon_k.
     \end{align}
     Similarly
    \begin{align}\label{eq:convergence hatphiAphiB}
        {\norm{\left(\Phi^k_A\otimes\Phi_B^k\right)[\hat{\rho}_{AB}]-(\hat{\rho}_{AB}\oplus 0)}_1=2\epsilon_k\longrightarrow0.}
    \end{align}
     {Since $\epsilon_k=1-\sum_{i<k}c_{ii}\to0$, we obtain}
    \begin{align}\label{eq:convergence phiAphiB}
        &\lim_{k\to\infty}\norm{\left(\Phi^k_A\otimes\Phi_B^k\right)[\rho_{AB}]-(\rho_{AB}\oplus 0)}_1=0,\qquad\lim_{k\to\infty}\norm{\left(\Phi^k_A\otimes\Phi_B^k\right)[\hat{\rho}_{AB}]-(\hat{\rho}_{AB}\oplus 0)}_1=0.        
    \end{align}
    With this in hand, we obtain
    \begin{align}
        D(\rho_{AB}\Vert\hat{\rho}_{AB})&=D(\rho_{AB}\oplus 0\Vert\hat{\rho}_{AB}\oplus 0)\nonumber\\
        &\leq\liminf_{k\to\infty}{D\left(\left(\Phi^k_A\otimes\Phi_B^k\right)[\rho_{AB}]\Vert\left(\Phi^k_A\otimes\Phi_B^k\right)[\hat{\rho}_{AB}]\right)}\nonumber\\
        &{\leq R(A:B)_\rho},
    \end{align}
    where the first inequality follows from the joint lower semicontinuity of relative entropy~\cite[Lemma~4]{LS23} and the second from~\eqref{eq:Rlowerk}. Combining this with~\eqref{eq:Rupper} proves the claim.
\end{proof}
\end{document}